\documentclass[10pt]{article}
\usepackage{amsthm,amsmath,amssymb}
\usepackage{microtype}
\usepackage{paralist}
\usepackage{charter,eulervm}
\usepackage{subfig}
\usepackage[pdftex,breaklinks,colorlinks,
linkcolor=blue,
citecolor=blue,
urlcolor=blue]{hyperref}
\usepackage{tkz-graph}
\usetikzlibrary{arrows,shapes,decorations.pathmorphing}
\graphicspath{{./figs/}}

\newtheorem{theorem}{Theorem}[section]
\newtheorem{lemma}[theorem]{Lemma}

\newtheorem{proposition}[theorem]{Proposition}

\newtheorem{definition}{Definition}
\def\boxit#1{\vbox{\hrule\hbox{\vrule\kern4pt
  \vbox{\kern1pt#1\kern1pt}
\kern2pt\vrule}\hrule}}
\newcommand{\keywords}[1]{\bigskip \par\noindent
{\small{\em Keywords\/}: #1}}

\newcommand{\badgraph}{minimal forbidden induced subgraph}
\newcommand{\hv}[1]{\ensuremath{N_H[#1]}}
\newcommand{\nhcag}{normal Helly circular-arc graph}
\newcommand{\ce}[1]{\ensuremath{{\mathtt{cp}(#1)}}}
\newcommand{\cce}[1]{\ensuremath{{\mathtt{ccp}(#1)}}}
\newcommand{\lp}[1]{\ensuremath{{\mathtt{lp}(#1)}}}
\newcommand{\rp}[1]{\ensuremath{{\mathtt{rp}(#1)}}}
\newcommand{\head}[1]{\ensuremath{{\mathtt{last}(#1)}}}
\newcommand{\tail}[1]{\ensuremath{{\mathtt{first}(#1)}}}
\newcommand{\stpath}[2]{($#1$, $#2$)-path}

\newcommand{\ec}{\ensuremath{E_{\text{c}}}}
\newcommand{\ecc}{\ensuremath{E_{\text{cc}}}}
\newcommand{\oc}{\ensuremath{T_{\text{c}}}}
\newcommand{\oo}{\ensuremath{T}}
\newcommand{\occ}{\ensuremath{T_{\text{cc}}}}
\newcommand{\og}[1]{\ensuremath{\phi(#1)}}
\newcommand{\comment}[1]{\hfill $\setminus\!\!\setminus$ {\em #1}}

\title{Forbidden Induced Subgraphs of Normal Helly \\Circular-Arc
  Graphs: Characterization and Detection\thanks{Preliminary results of
    this paper appeared in the proceedings of SBPO 2012
    \cite{grippo-12-cag-without-dominating-triple} and FAW 2014
    \cite{cao-14-recognizing-nhcag}.}}

\author{Yixin Cao\thanks{Institute for Computer Science and Control,
    Hungarian Academy of Sciences.  Email:
    \href{mailto:yixin@sztaki.hu}{\tt yixin@sztaki.hu}.  Supported by
    the European Research Council (ERC) under the grant 280152 and the
    Hungarian Scientific Research Fund (OTKA) under the grant
    NK105645. }  \and Luciano N. Grippo\thanks{Instituto de Ciencias,
    Universidad Nacional de General Sarmiento, Los Polvorines, Buenos
    Aires, Argentina.
    Email:\href{mailto:lgrippo@ungs.edu.ar}{lgrippo@ungs.edu.ar},
    \href{mailto:msafe@ungs.edu.ar}{msafe@ungs.edu.ar}.  Partially
    supported by CONICET PIP 11220120100450CO and ANPCyT PICT
    2012-1324 grants.}
  \addtocounter{footnote}{-1}
  \and Mart\'in D. Safe\footnotemark
 }

\date{\today}

\begin{document}
\maketitle
\tikzstyle{corner}  = [fill=blue,inner sep=2.5pt]
\tikzstyle{special} = [fill=black,circle,inner sep=2pt]
\tikzstyle{vertex}  = [fill=black,circle,inner sep=2pt]
\tikzstyle{original edge} = [thick,-,blue,dashed]
\tikzstyle{forbidden edge} = [dashed,-,red]
\tikzstyle{edge}    = [draw,thick,-]
\tikzstyle{at edge} = [draw,ultra thick,-,red]

\begin{abstract}
  A normal Helly circular-arc graph is the intersection graph of arcs
  on a circle of which no three or less arcs cover the whole circle.
  Lin, Soulignac, and Szwarcfiter [Discrete Appl.\ Math.\ 2013]
  characterized circular-arc graphs that are not normal Helly
  circular-arc graphs, and used it to develop the first recognition
  algorithm for this graph class.  As open problems, they ask for the
  forbidden induced subgraph characterization and a direct recognition
  algorithm for normal Helly circular-arc graphs, both of which are
  resolved by the current paper.  Moreover, when the input is not a
  normal Helly circular-arc graph, our recognition algorithm finds in
  linear time a minimal forbidden induced subgraph as certificate.
 \end{abstract}
 \keywords{certifying algorithms, holes, interval models, (minimal)
   forbidden induced subgraphs, (normal, Helly) circular-arc models.}

\section{Introduction} 
This paper will be only concerned with undirected and simple graphs.
A graph is a \emph{circular-arc graph} if its vertices can be assigned
to arcs on a circle such that two vertices are adjacent if and only if
their corresponding arcs intersect.  Such a set of arcs is called a
\emph{circular-arc model} of this graph.  If there is some point on
the circle that is not covered by any arc in the model, then the graph
is an \emph{interval graph}, and it can also be represented by a set
of intervals on the real line, which is called an \emph{interval
  model}.  Circular-arc graphs and interval graphs are two of the most
famous intersection graph classes, and both have been studied
intensively for decades.  However, in contrast to interval graphs, our
understanding of circular-arc graphs is far limited, and to date some
fundamental problem remains unsolved.

\begin{figure*}[h]
  \centering\footnotesize
  \subfloat[long claw]{\label{fig:long-claw}
    \includegraphics{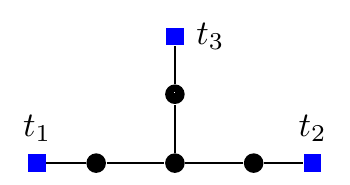} 
  }
  $\,$
  \subfloat[whipping top]{\label{fig:whipping-top}
    \includegraphics{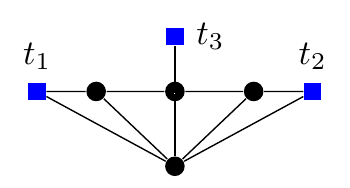} 
  }
  $\,$
  \subfloat[\dag]{\label{fig:net}
    \includegraphics{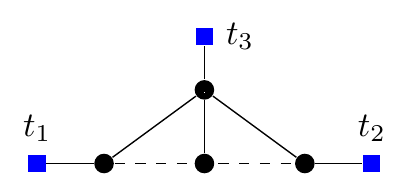} 
  }
  $\,$
  \subfloat[\ddag]{\label{fig:tent}
    \includegraphics{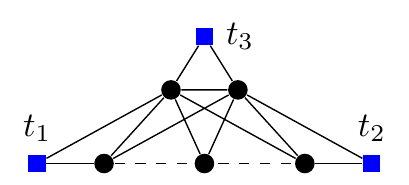} 
  }
  \caption{Chordal minimal forbidden induced graphs.}
  \label{fig:at}
\end{figure*}

One fundamental combinatorial problem on a graph class is its
characterization by forbidden induced subgraphs.  For example,
Lekkerkerker and Boland \cite{lekkerkerker-62-interval-graphs} showed
in 1962 that a graph is an interval graph if and only if it contains
neither a hole (i.e., a induced cycle of length at least four) nor any
graph in Fig.~\ref{fig:at} as an induced subgraph.  Recall that holes
are the forbidden induced subgraphs of \emph{chordal graphs}, which
are the intersection of subtrees of a tree.  In contrast, since it was
first asked by Hadwiger et
al.~\cite{hadwiger-64-combinatorial-geometry} in 1964, all efforts
attempting to characterize circular-arc graphs by forbidden induced
subgraphs have succeeded only partially.  Tucker made the most
significant contribution to the study of the class of circular-arc
graphs and its subclasses, which includes the forbidden induced
subgraph characterizations of both unit circular-arc graphs (i.e., a
graph with a circular-arc model where every arc has the same length)
and proper circular-arc graphs (i.e., a graph with a circular-arc
model where no arc properly contains another)
\cite{tucker-74-structures-cag}; we will see more later.  There is a
similar line of research for other proper subclasses of circular-arc
graphs, which aims at determining their forbidden induced subgraphs or
some other kinds of obstructions; for this we refer to the surveys of
Lin and Soulignac \cite{lin-09-cag-and-subclasses} and Dur\'an et
al.~\cite{duran-14-survey} and references therein.

One fundamental algorithmic problem on a graph class is its
recognition, i.e., to efficiently decide whether a given graph belongs
to this class or not.  For intersection graph classes, all
{recognition algorithms} known to the authors provide an intersection
model when the membership is asserted.  Most of them, on the other
hand, simply return ``NO'' for non-membership, while one might also
want some verifiable \emph{certificate} for some reason
\cite{mcconnell-11-survey-certifying-algorithms}.  A recognition
algorithm is \emph{certifying} if it provides both positive and
negative certificates.  There are different forms of negative
certificates, while a minimal forbidden induced subgraph is arguably
the simplest and most preferable of them
\cite{heggernes-07-certifying-fis}.  Kratsch et
al.~\cite{kratsch-06-certifying-interval-and-permutation} reported a
certifying recognition algorithm for interval graphs, which in linear
time returns either an interval model of an interval graph or a
forbidden induced subgraph for a non-interval graph.  Although its
returned forbidden induced subgraph is not necessarily minimal, a
minimal one can be easily retrieved from it (see also
\cite{lindzey-13-find-forbidden-subgraphs} for another approach).
Likewise, a hole can be detected from a non-chordal graph in linear
time \cite{tarjan-84-chordal-recognition}.  On the other hand,
although a circular-arc model of a circular-arc graph can be produced
in linear time \cite{mcconnell-03-recognition-cag}, it remains a
challenging open problem to find a negative certificate for a
non-circular-arc graph.

The complication of circular-arc graphs may be attributed to two
special intersection patterns of circular-arc models that are not
possible in interval models.  The first is two arcs intersecting in
both ends, and a {circular-arc model} is called \emph{normal} if no
such pair exists.  The second is a set of arcs intersecting pairwise
but containing no common point, and a {circular-arc model} is called
\emph{Helly} if no such set exists.  Normal and Helly circular-arc
models are precisely those without three or less arcs covering the
whole circle
\cite{mckee-03-restricted-cag,lin-13-nhcag-and-subclasses}.  A graph
that admits such a model is called a \emph{normal Helly circular-arc
  graph}.  In particular, all interval graphs are normal Helly
circular-arc graphs.  

\begin{figure}[h!]
  \centering
  \subfloat[A circular-arc graph $G$]{\label{fig:example-graph}
    \includegraphics{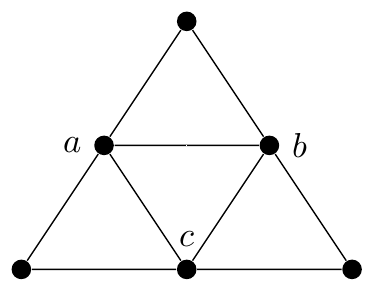} 
  }
  \qquad
  \subfloat[{A normal model of $G$}]{\label{fig:normal-model}
    \includegraphics{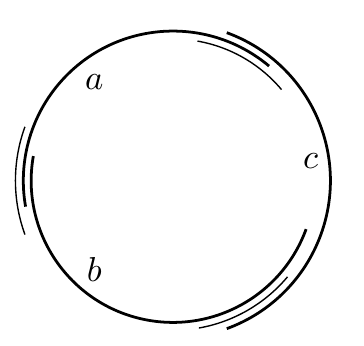} 
  }
  \qquad
  \subfloat[{A Helly model of $G$}]{\label{fig:helly-model}
    \includegraphics{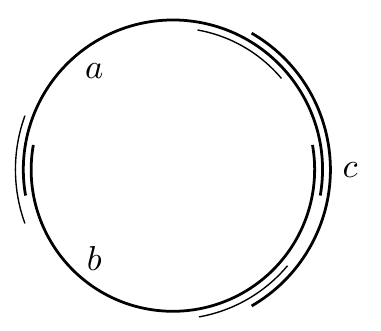} 
}
  \caption{a non-\nhcag\ and its circular-arc models}
  \label{fig:normal-and-helly}
\end{figure}

A word of caution is worth on the definition of \nhcag s.  One graph
might admit both a normal circular-arc model and a Helly circular-arc
model but not a normal and Helly circular-arc model.  For example, see
the graph and its models in Fig.~\ref{fig:normal-and-helly}.  The fact
that the model of Fig.~\ref{fig:normal-model} (resp.,
Fig.~\ref{fig:helly-model}) is not Helly (resp., normal) can be
evidenced by arc set $\{a,b,c\}$ (resp., $\{a,b\}$).  One may want to
verify that arranging a normal and Helly circular-arc model for this
graph is out of the question.  This example convinces us that the set
of \nhcag s is {\em not} equivalent to the intersection of normal
circular-arc graphs and Helly circular-arc graphs, but a proper subset
of it.

Let us mention some previous work related to normal Helly circular-arc
graphs.  Tucker \cite{tucker-75-coloring-cag} gave an algorithm that
outputs a proper coloring of any given normal Helly circular-arc graph
using at most $3\omega/2$ colors, where $\omega$ denotes the size of a
maximum clique.  Note that by the Helly property, $\omega$ is
equivalent to the maximum number of arcs covering a single point on
the circle.  This is tight as any odd hole, which has $\omega = 2$ and
needs at least three colors, is a \nhcag.  In the study of convergence
of circular-arc graphs under the clique operator, Lin et
al.~\cite{lin--10-clique-operator-cag} observed that normal Helly
circular-arc graphs arose naturally.  They then
\cite{lin-13-nhcag-and-subclasses} undertook a systematic study of
normal Helly circular-arc graphs as well as its subclass.  Their
results include a partial characterization of \nhcag s by forbidden
induced subgraph (more specifically, those restricted to Helly
circular-arc graphs), and a linear-time recognition algorithm (by
calling a recognition algorithm for circular-arc graphs).  As open
problems, they ask for determining the remaining minimal forbidden
induced subgraphs, and designing a direct recognition algorithm, both
of which are resolved by the current paper.

The first main result of this paper is a complete characterization of
\nhcag s by forbidden induced subgraphs.  A wheel (resp., $C^*$)
comprises a hole and another vertex completely adjacent (resp.,
nonadjacent) to it.

\begin{theorem}\label{thm:characterization}
A graph is a normal Helly circular-arc graph if and only if it
contains no $C^*$, wheel, or any graph depicted in Figs.~\ref{fig:at}
and \ref{fig:normal-helly}.
\end{theorem}
\begin{figure*}[t]
  \centering
  \subfloat[$K_{2,3}$]{\label{fig:long-claw-1}
    \includegraphics{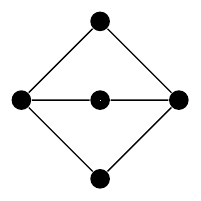} 
  }
  $\quad$
  \subfloat[twin-$C_5$]{\label{fig:g-2}
    \includegraphics{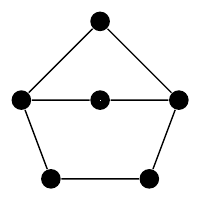} 
  }
  $\quad$
  \subfloat[domino]{\label{fig:domino}
    \includegraphics{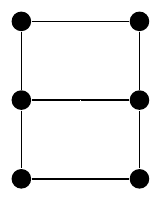} 
  }
  $\quad$  
  \subfloat[$\overline{C_6}$]{\label{fig:complement-c-6}
    \includegraphics{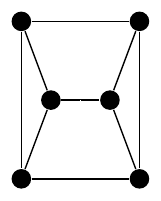} 
  }
  $\quad$  
  \subfloat[FIS-1]{\label{fig:fis-1}
    \includegraphics{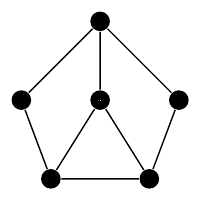} 
  }
  $\quad$
  \subfloat[FIS-2]{\label{fig:fis-2}
    \includegraphics{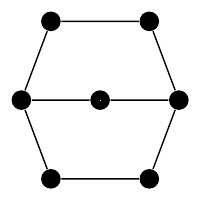} 
  }
  \caption{Non-chordal and finite minimal forbidden induced graphs.}
  \label{fig:normal-helly}
\end{figure*}

It is easy to use the definition to verify that a normal Helly
circular-arc graph is chordal if and only if it is an interval graph.
An interval model is always a normal and Helly circular-arc model, but
an interval graph might have circular-arc model that is neither normal
nor Helly, e.g., consider $K_4$.  For non-chordal graphs we have:
\begin{proposition}[\cite{mckee-03-restricted-cag,lin-13-nhcag-and-subclasses}]
  \label{thm:always-noraml-and-helly}
  If a normal Helly circular-arc graph $G$ is not chordal, then every
  circular-arc model of $G$ is normal and Helly.
\end{proposition}
These observations inspire us to recognize \nhcag s as follows.  If
the input graph is chordal, it suffices to check whether it is an
interval graph.  Otherwise, we try to build a circular-arc model of
it, and if we succeed, verify whether the model is normal and Helly.  Lin
et al.~\cite{lin-13-nhcag-and-subclasses} showed that this approach
can be implemented in linear time.  Moreover, if there exists a set of
at most three arcs covering the circle, then their algorithm returns
it as a certificate.
This algorithm, albeit conceptually simple, suffers from twofold
weakness.  First, it needs to call some recognition algorithm for
circular-arc graphs, while all known algorithms are extremely
complicated.  Second, it is very unlikely to deliver a negative
certificate in general.  

The second main result of this paper is the following direct
certifying algorithm for recognizing \nhcag s, which would be
desirable for both efficiency and the detection of negative
certificates.  From now on, unless otherwise stated, whenever we refer
to a ``minimal forbidden induced subgraph'' it should be understood a
minimal forbidden induced subgraph for the class of normal Helly
circular-arc graphs.  We use $n := |V(G)|$ and $m := |E(G)|$
throughout.
\begin{theorem}\label{thm:certifying-algorithm}
  There is an $O(n+m)$-time algorithm that given a graph $G$, either
  constructs a normal and Helly circular-arc model of $G$, or finds a
  minimal forbidden induced subgraph of $G$.
\end{theorem}

It is clear that each graph specified in
Theorem~\ref{thm:characterization} is a minimal forbidden induced
subgraph.  First, every graph in Fig.~\ref{fig:at} is chordal but
non-interval graph, and thus cannot be a \nhcag.  Second, a $C^*$ is
not a circular-arc graph, while a wheel cannot be arranged without
three or less arcs covering the circle.  Third, every graph in
Fig.~\ref{fig:normal-helly} has only a small number of vertices and
can be easily checked.  Therefore, to prove
Theorem~\ref{thm:characterization}, it suffices to show that a graph
containing none of them is a \nhcag.  That fact was actually proved in
\cite{grippo-12-cag-without-dominating-triple}, but the resulting
proof of Theorem 1.1 given there does not provide a linear-time
procedure to find the corresponding forbidden induced subgraphs when
the graph is not a normal Helly circular-arc graph.  Since the
algorithm we use to prove Theorem~\ref{thm:certifying-algorithm}
always finds such a subgraph in this case,
Theorem~\ref{thm:characterization} follows from the correctness proof
of our algorithm as a corollary.

Let us briefly discuss the basic idea behind the way we deal with a
non-chordal graph $G$.  If $G$ is a \nhcag, then for any vertex $v$ of
$G$, both $N[v]$ and its complement induce nonempty interval
subgraphs.  The main technical difficulty is how to combine interval
models for them to make a circular-arc model of $G$.  For this purpose
we build an auxiliary graph $\mho(G)$ by taking two identical copies
of $N[v]$ and appending them to the two ends of $G - N[v]$
respectively.  {The shape of symbol $\mho$ is a good hint for
  understanding the structure of the auxiliary graph.} We show that
$\mho(G)$ is an interval graph and more importantly, a circular-arc
model of $G$ can be produced from an interval model of $\mho(G)$.  On
the other hand, if $G$ is not a \nhcag, then $\mho(G)$ cannot be an
interval graph.  In this case we use the following procedure to obtain
a minimal forbidden induced subgraph of $G$.
\begin{theorem}\label{thm:negative-certificate}
  Given a minimal non-interval induced subgraph of $\mho(G)$, we can
  in $O(n+m)$ time find a \badgraph\ of $G$.
\end{theorem}
The crucial idea behind our certifying algorithm is a novel
correlation between normal Helly circular-arc graphs and interval
graphs, which can be efficiently used for algorithmic purpose.  This
was originally proposed in the detection of small forbidden induced
subgraph of interval graphs \cite{cao-14-almost-interval-recognition},
i.e., the opposite direction of the current paper.  In particular, in
\cite{cao-14-almost-interval-recognition} we have used a similar
definition of the auxiliary graph and pertinent observations.
However, the main structures and the procedures for the detection of
forbidden induced subgraphs divert completely.  For example, the most
common forbidden induced subgraphs in
\cite{cao-14-almost-interval-recognition} are $4$- and $5$-holes,
which, however, are allowed in normal Helly circular-arc graphs.  This
means that the interaction between $N[v]$ and $G - N[v]$ are far more
subtle, and thus the detection of \badgraph s in the current paper is
significantly more complicated than that of
\cite{cao-14-almost-interval-recognition}.

\section{The recognition algorithm}\label{sec:recognition}
All graphs are stored as adjacency lists.  We use the customary
notation $v\in G$ to mean $v\in V(G)$, and $u\sim v$ to mean $uv\in
E(G)$.  The \emph{degree} of a vertex $v$ is defined by $d(v) :=
|N(v)|$, where $N(v)$, called the \emph{neighborhood} of $v$,
comprises all vertices $u$ such that $u\sim v$.  The \emph{closed
  neighborhood} of $v$ is defined by $N[v] := N(v) \cup \{v\}$.  For a
vertex set $U$, its closed neighborhood and neighborhood are defined
by $N[U] := \bigcup_{v \in U} N[v]$ and $N(U) := N[U] \backslash U$,
respectively.  Exclusively concerned with induced subgraphs, we use
$F$ to denote both a subgraph and its vertex set.

Consider a circular-arc model $\cal A$.  If every point of the circle
is contained in some arc in $\cal A$, then we can find an
inclusion-wise minimal set $X$ of arcs that cover the entire circle.
If $\cal A$ is normal and Helly, then $X$ consists of at least four
vertices and thus corresponds to a hole.  Therefore, a \nhcag\ $G$ is
chordal if and only if it is an interval graph, for which it suffices
to call the algorithms of
\cite{kratsch-06-certifying-interval-and-permutation,lindzey-13-find-forbidden-subgraphs}.
We are hence focused on graphs that are not chordal.  We call the
algorithm of Tarjan and Yannakakis
\cite{tarjan-84-chordal-recognition} to detect a hole $H$.
\begin{proposition}\label{lem:fundamental}
  Let $H$ be a hole of a circular-arc graph $G$.  In any circular-arc
  model of $G$, the union of arcs for $H$ covers the whole circle.
  In other words, $N[H] = V(G)$.
\end{proposition}

The indices of vertices in the hole $H = (h_0 h_1 \ldots h_{\vert
  H\vert-1} h_0)$ should be understood as modulo $|H|$, e.g., $h_{-1}
= h_{|H|-1}$.  By Proposition~\ref{lem:fundamental}, every vertex
should have neighbors in $H$.  We use $N_H[v]$ as a shorthand for
$N[v] \cap H$, regardless of whether $v \in H$ or not.  We start from
characterizing $N_H[v]$ for every vertex $v$: we specify some
forbidden structures not allowed to appear in a \nhcag, and more
importantly, we show how to find a minimal forbidden induced subgraph
if one of these structures exists.  The fact that they are forbidden
can be easily seen from the definition of normal and Helly and
Proposition~\ref{lem:fundamental}, and hence the proofs given below
will focus on the detection of \badgraph s.
\begin{lemma}\label{lem:non-consecutive}
  For every vertex $v$, we can in $O(d(v))$ time find either a proper
  sub-path of $H$ induced by $N_H[v]$, or a \badgraph.
\end{lemma}
\begin{proof}
  We pre-allocate a list {\sf IND} of $d(v)$ slots, initially all
  empty.  For each neighbor of $v$, if it is $h_i$, then add $i$ into
  the next empty slot of {\sf IND}.  After all neighbors of $v$ have
  been checked, we shorten {\sf IND} by removing empty slots from the
  end, which leaves $|N_H[v]|$ slots.  If $|N_H[v]|$ is $0$ or $|H|$,
  then we return $H$ and $v$ as a $C^*$ or wheel.  In the remaining
  case, $N_H[v]$ is a nonempty and proper subset of $H$.  We radix
  sort {\sf IND}; let $p$ and $q$ be its first and last elements
  respectively.

  Starting from the first element, we traverse {\sf IND} to the end
  for the first $i$ such that ${\sf IND}[i+1] > {\sf IND}[i] + 1$.  If
  no such $i$ exists, then we return ($h_p\cdots h_q$) as the path.
  In the remaining cases, we may assume that we have found the $i$;
  let $p_1 := {\sf IND}[i]$ and $p_2 := {\sf IND}[i+1]$.  We continue
  to traverse from $i+1$ to the end of {\sf IND} for the first $j$
  such that ${\sf IND}[j+1] > {\sf IND}[j] + 1$.  This step has three
  possible outcomes:
  \begin{inparaenum}[(1)]
  \item if $j$ is found, then $p_3 := {\sf IND}[j]$ and $p_4 := {\sf
      IND}[j+1]$;
  \item if no such $j$ is found, and at least one of $q<|H|-1$ and
    $p>0$ holds, then $p_3 := q$ and $p_4 :=  p + |H|$; and
  \item otherwise ($p=0$, $q=|H|-1$, and $j$ is not found).
  \end{inparaenum}
  In the third case, we return ($h_{p_2}\cdots h_{|H| - 1} h_0 \cdots
  h_{p_1}$) as the path induced by \hv{v}.  In the first two cases,
  $p_3$ and $p_4$ are defined, and $p_4 > p_3 +1$.  In other words, we
  have two nontrivial sub-paths, ($h_{p_1} h_{p_1+1}\dots h_{p_2}$)
  and ($h_{p_3} h_{p_3+1}\dots h_{p_4}$), of $H$ such that $v$ is
  adjacent to their ends but none of their inner vertices.

  If $p_2 - p_1 >3$, then we return ($v h_{p_3} h_{p_3+1}\dots h_{p_4}
  v$) and $h_{p_1 + 2}$ as a $C^*$.  Likewise, if $v\not\sim h_{\ell}$
  for some $\ell$ with $p_2 + 1<\ell<p_1 - 1 + |H|$, then we return
  ($v h_{p_1} h_{p_1+1}\dots h_{p_2} v$) and $h_{\ell}$ as a $C^*$;
  note this must hold true when $v$ is adjacent to both $h_{p_1 - 1}$
  and $h_{p_2 + 1}$.  Hence we may assume $2\le p_2 - p_1\le 3$, and
  without loss of generality, $v\not\sim h_{p_1 - 1}$.  

  If $p_2 - p_1 = 2$, then we return
  \begin{inparaenum}[(\itshape 1\upshape)]
  \item $H\cup \{v\}$ as a $K_{2,3}$ when $|H|=4$;
  \item $H\cup \{v\}$ as a twin-$C_5$ when $|H|=5$ and $|N_H[v]| = 2$;
  \item $H\cup \{v\}$ as an FIS-1 when $|H|=5$ and $|N_H[v]| = 3$; or
  \item $\{h_{p_1-2}, h_{p_1 - 1} \cdots, h_{p_2},v\}$ as a domino
    when $|H|>5$.
  \end{inparaenum}
  Otherwise, $p_2 - p_1 = 3$, and we return
  \begin{inparaenum}[(\itshape 1\upshape)]
  \item $H\cup \{v\}$ as a twin-$C_5$ when $|H|=5$;
  \item $H\cup \{v\}$ as an FIS-2 when $|H|=6$ and $v\not\sim
    h_{p_2+1}$;
  \item ($v h_{p_1} h_{p_1-1} h_{p_1-2} v$) and $h_{p_2-1}$ as a $C^*$
    when $|H|=6$ and $v\sim h_{p_2+1}$; or
  \item ($v h_{p_1} h_{p_1-1} h_{p_1-2} v$) and $h_{p_2-1}$ as a $C^*$ when
    $|H|>6$.
  \end{inparaenum}

  The construction of {\sf IND} takes $O(d(v))$ time.  In the same
  time we can traverse it to find indices $p_1,p_2,p_3,p_4$.  The rest
  uses constant time.  This concludes the time analysis and completes
  the proof.
\end{proof}

We designate the ordering $h_0, h_1, h_2, \cdots$ of traversing $H$ as
\emph{clockwise}, and the other \emph{counterclockwise}.  In other
words, edges $h_0 h_{1}$ and $h_0 h_{-1}$ are clockwise and
counterclockwise from $h_0$, respectively.  Now let $P$ be the path
induced by $\hv{v}$.  We can assign a direction to $P$ in accordance
to the direction of $H$, and then we have clockwise and
counterclockwise ends of $P$.  For technical reasons, we assign
canonical indices to the ends of the path $P$ as follows.
\begin{definition}
  For each vertex $v\in G$, we denote by \tail{v} and \head{v} the
  indices of the counterclockwise and clockwise, respectively, ends of
  the path induced by \hv{v} in $H$ satisfying
  \begin{itemize}
  \item $- |H| < \tail{v} \le 0\le \head{v} < |H|$ {if } $h_0\in
    N_H[v]$; or
  \item $0< \tail{v}\le \head{v}< |H|$, {otherwise}.
  \end{itemize}
\end{definition}
It is possible that $\head{v} = \tail{v}$, when $|\hv{v}| = 1$.  In
general, $\head{v} - \tail{v} = |\hv{v}| - 1$, and $v=h_i$ or $v\sim
h_i$ for each $i$ with $\tail{v}\le i\le \head{v}$.  The indices
$\tail{v}$ and $\head{v}$ can be easily retrieved from
Lemma~\ref{lem:non-consecutive}, with which we can check the adjacency
between $v$ and any vertex $h_i\in H$ in constant time.  Now consider
the neighbors of more than one vertices in $H$.
\begin{lemma}\label{lem:non-consecutive-2}
  Given a pair of adjacent vertices $u,v$ such that $N_H[u]$ and
  $N_H[v]$ are disjoint, then in $O(n+m)$ time we can find a
  \badgraph.
\end{lemma}
\begin{proof}
  Clearly, neither of $u$ and $v$ can be in $H$.  We may assume both
  $N_H[u]$ and $N_H[v]$ induce proper sub-paths; otherwise we can call
  Lemma~\ref{lem:non-consecutive}.  They partition $H$ into four
  sub-paths, two of which are induced by $N_H[u]$ and $N_H[v]$.
  Denote by $P_1$ and $P_2$ the other two sub-paths; their ends are
  adjacent to $u$ and $v$ respectively, while their inner vertices, if
  any, are adjacent to neither $u$ nor $v$..

  Assume first that both $P_1$ and $P_2$ are of length $1$, then $|H|
  = |N_H[u]| + |N_H[v]|$.  If $u$ is adjacent to a single vertex $h_i$
  in $H$, (noting that $|N_H[v]|\ge 3$,) then we return $\{h_i,
  h_{i-1}, h_{i+1}, u, v\}$ as a $K_{2,3}$.  A symmetric argument
  applies when $|N_H[v]|= 1$.  If $|N_H[u]| = |N_H[v]| = 2$, then we
  return $H\cup \{u,v\}$ as a $\overline{C_6}$.  It must be in some
  case above if $|H| = 4$, and henceforth we assume $|H| > 4$.  If $u$
  is adjacent to only $h_i$ and $h_{i+1}$ in $H$, (noting that
  $|N_H[v]|\ge 3$,) then we return $\{h_{i-1}, h_i, h_{i+1}, h_{i+2},
  u, v\}$ as an FIS-1.  A symmetric argument applies when $|N_H[v]|=
  2$.  Now that both $|N_H[u]|$ and $|N_H[v]|$ are at least $3$, we
  return $P_1\cup P_2\cup\{u,v\}$ as a domino.

  Assume now that, without loss of generality, $P_2$ is nontrivial.
  We can return ($v u P_1 v$)$+h_{\head{v} + 1}$ (when both $|N_H[u]|
  > 1$ and $|N_H[v]| > 1$) or ($v u P_1 v$) and $h_{\head{v} + 2}$ (when
  the length of $P_1$ is longer than $3$) as a $C^*$.  A symmetric
  argument applies when $\tail{v} - \head{u} > 3$.  In the remaining
  cases, we assume without loss of generality, $|N_H[u]| = 1$, and
  both paths $P_1$ and $P_2$ contain at most $4$ vertices.
  Consequently, $|H\setminus N_H[v]|\le 5$.

  If $P_1$ is also nontrivial, then $|H\setminus N_H[v]|\ge 3$.  We
  return $\{u, v, h_{\tail{v}-1}, h_{\head{v}+1}\}\cup N_H[v]$ as a
  \dag\ when $|N_H[v]| > 1$.
  Now that $|N_H[v]|= 1$, then $|H|\le 6$, and we return
  \begin{inparaenum}[(\itshape 1\upshape)]
  \item $(H\setminus N_H[v])\cup \{u, v\}$ as a long claw when $|H|=
    6$;
  \item  $H\cup\{u,v\}$ as a twin-$C_5$ when $|H|= 4$; or
  \item $H\cup\{u,v\}$ as a FIS-2 when $|H|= 5$.
  \end{inparaenum}
  In the final case, $P_1$ is trivial but $P_2$ is nontrivial, which
  means that neither $u$ nor $v$ is adjacent to $h_{\tail{u}-1}$.  If
  $v\sim h_{\tail{u}-2}$, then we return $\{h_{\tail{u}-2},
  h_{\tail{u}-1}, h_{\tail{u}}, h_{\tail{u}+1}, u, v\}$ as
  \begin{inparaenum}[(\itshape 1\upshape)]
  \item an FIS-1 when $|H| = 4$; or
  \item a twin-$C_5$ when $|H| > 4$.
  \end{inparaenum}
  If $v\not\sim h_{\tail{u}-2}$, then we return
  \begin{inparaenum}[(\itshape 1\upshape)]
  \item $H\cup\{u,v\}$ as a domino when $|H| = 4$; or
  \item ($u v h_1 h_0 u$) and $h_{-2}$ as a $C^*$ when $|H| > 4$.
  \end{inparaenum}
  This procedure enters only one case, which is decided only by
  $N_H[u]$ and $N_H[v]$.  Therefore, it can be done in $O(n+m)$ time.
\end{proof}

\begin{lemma}\label{lem:non-helly}
  Given a set $U$ of two or three pairwise adjacent vertices such that
  \begin{enumerate}[1)]
  \item $\bigcup_{u\in U} N_H[u] = H$; and 
  \item for every $u\in U$, each end of $N_H[u]$ is adjacent to at
    least two vertices in $U$,
  \end{enumerate}
  then we can in $O(n+m)$ time find a \badgraph.
\end{lemma}
\begin{proof}
  Consider first that $U$ contains only two vertices $v_1$ and $v_2$.
  The \stpath{h_{\tail{v_1}}}{h_{\head{v_1}}} whose inner vertices are
  nonadjacent to $v_1$ makes a hole with $v_1$.  This hole is
  completely adjacent to $v_2$, and thus we return a wheel.

  Consider then $U = \{v_1, v_2, v_3\}$.  We may assume that no two
  vertices of $U$ satisfy the condition of the lemma, as otherwise we
  are in the previous case.  Without loss of generality, assume that
  $h_{\head{v_1}}\in N[v_2]$, and then $h_{\tail{v_2}}\in N[v_1]$.
  The \stpath{h_{\tail{v_1}}}{h_{\head{v_2}}} whose inner vertices are
  adjacent to neither $v_1$ or $v_2$ makes a hole with $v_1$ and
  $v_2$.  By assumption, $v_3$ is adjacent to every vertex in the
  hole, and thus we return a wheel.
\end{proof}

Let $\oo := N[h_0]$ and $\overline{\oo} := V(G)\setminus \oo$.  As we
have alluded to earlier, we want to duplicate $\oo$ and append them to
different sides of $\overline{\oo}$.  Each edge between $v\in \oo$ and
$u\in \overline{\oo}$ will be carried by only one copy of $T$, and
this is determined by its direction specified as follows.  We may
assume that none of the Lemmas.~\ref{lem:non-consecutive},
\ref{lem:non-consecutive-2}, and \ref{lem:non-helly} applies to $v$
or/and $u$, as otherwise we can terminate the algorithm by returning
the forbidden induced subgraph found by them.  As a result, $u$ is
adjacent to either $\{h_{\tail{v}},\cdots, h_{-1}\}$ or $\{h_1,\cdots,
h_{\head{v}}\}$ but not both.  The edge $u v$ is said to be clockwise
from $\oo$ if $u\sim h_i$ for $1\le i\le {\head{v}}$, and
counterclockwise otherwise.  Let \ec\ (resp., \ecc) denote the set of
clockwise (resp., counterclockwise) edges from $\oo$, and let \oc\
(resp., \occ) denote the subsets of vertices of $\oo$ that are
incident to edges in \ec\ (resp., \ecc).  Note that $\{\ecc,\ec\}$
partitions edges between $\oo$ and $\overline{\oo}$, but a vertex in
$\oo$ might belong to both \occ\ and \oc, or neither of them.  We have
now all the details for the definition of the auxiliary graph
$\mho(G)$.

\begin{definition}
  The vertex set of $\mho(G)$ consists of $\overline{\oo}\cup L\cup
  R\cup \{w\}$, where $L$ and $R$ are distinct copies of $\oo$, i.e.,
  for each $v\in \oo$, there are a vertex $v^l$ in $L$ and another
  vertex $v^r$ in $R$, and $w$ is a new vertex distinct from $V(G)$.
  For each edge $u v\in E(G)$, we add to the edge set of $\mho(G)$
  \begin{itemize}
  \item an edge $u v$ if neither $u$ nor $v$ is in $\oo$;
  \item two edges $u^l v^l$ and $u^r v^r$ if both $u$ and $v$ are in
    $\oo$; or
  \item an edge $u v^l$ or $u v^r$ if $uv\in \ec$ or $uv\in \ecc$
    respectively ($v\in \oo$ and $u\in \overline\oo$).
  \end{itemize}
  Finally, we add an edge $w v^l$ for every $v\in \occ$.
\end{definition} 

\begin{lemma}\label{lem:construct-mho}
  The numbers of vertices and edges of $\mho(G)$ are upper bounded by
  $2 n$ and $2 m$ respectively.  Moreover, an adjacency list
  representation of $\mho(G)$ can be constructed in $O(n + m)$ time.
\end{lemma}
\begin{figure}[h!]
  \vspace*{-5mm}
  \setbox4=\vbox{\hsize28pc \noindent\strut
  \begin{quote}
  \vspace*{-5mm} \footnotesize

  {\sc input}: a graph $G$ and a hole $H$.
  \\
  {\sc output}: the auxiliary graph $\mho(G)$ or a forbidden induced
  subgraph of $G$.
  \\[1ex]
  0 \hspace*{2ex} {\bf for each} $v\in V(G)$ {\bf do} compute \tail{v}
  and \head{v};
  \\
  1 \hspace*{2ex} {\bf for each} $v\in \oo$ {\bf do}
  \\
  1.1 \hspace*{3ex} add vertices $v^l$ and $v^r$;
  \\
  1.2 \hspace*{3ex} {\bf for each} $u\in N(v)$ {\bf do}
  \\
  1.2.1 \hspace*{5ex} {\bf if} $u\in \oo$ {\bf then} add $u^l$ to
  $N(v^l)$ and $u^r$ to $N(v^r)$;
  \\
  1.2.2 \hspace*{5ex} {\bf else if} $u$ is not marked {\bf then} mark
  $u$ and put it into $N(\oo)$;
  \\
  2 \hspace*{2ex} {\bf for each} $u\in N(\oo)$ {\bf do}
  \\
  2.1 \hspace*{3ex} {\bf for each} $v\in N(u)$ {\bf do}
  \comment{$\tail{v}\le 0\le \head{v}$.}
  \\
  2.1.1 \hspace*{5ex} {\bf if} $v\not\in \oo$ {\bf then goto} 2.1;
  \comment{henceforth $v\in X$.}
  \\
  2.1.2 \hspace*{5ex} {\bf if} $\tail{v}= 0= \head{v}$ {\bf then
    return} a forbidden induced subgraph;
  \\
  2.1.3 \hspace*{5ex} {\bf if} $\head{v} = 0$ {\bf then} replace $v$
  by $v^l$ in $N(u)$ and add $u$ to $N(v^l)$; \comment{$u v\in \ecc$.}
  \\
  2.1.4 \hspace*{5ex} {\bf if} $\tail{v} = 0$ {\bf then} replace $v$
  by $v^r$ in $N(u)$ and add $u$ to $N(v^r)$; \comment{$u v\in \ec$.}
  \\
  \hspace*{5ex} $\setminus\!\!\setminus$ in the remaining cases
  $N[v]\cap V(H) = \{h_{-1}, h_0, h_1\}$ and $|N[u] \cap \{h_{-1},
  h_1\}| = 1$.
  \\
  2.1.5 \hspace*{5ex} {\bf if} $\head{u} = |H| - 1$ {\bf then} replace
  $v$ by $v^l$ in $N(u)$ and add $u$ to $N(v^l)$; \comment{$u v\in
    \ecc$.}
  \\
  2.1.6 \hspace*{5ex} {\bf if} $\tail{u} = 1$ {\bf then} replace $v$
  by $v^r$ in $N(u)$ and add $u$ to $N(v^r)$; \comment{$u v\in \ec$.}
  \\
  2.1.7 \hspace*{5ex} {\bf if} $u v\in \ecc$ and $v$ is not marked as
  \occ\ {\bf then} mark $v$ and put it into \occ;
  \\
  2.1.8 \hspace*{5ex} {\bf if} $u v\in \ec$ and $v$ is not marked as
  \oc\ {\bf then} mark $v$ and put it into \oc;
  \\
  3 \hspace*{2ex} add vertex $w$;
  \\
  4 \hspace*{2ex} {\bf for each} $v\in \occ$ {\bf do} put $w$ into
  $N(v^l)$ and $v^l$ into $N(w)$;
  \\
  5 \hspace*{2ex} remove $\oo$.
\end{quote} \vspace*{-3mm} \strut} $$\boxit{\box4}$$
\vspace*{-7mm}
\caption{Procedure for constructing $\mho(G)$
  (Lemma~\ref{lem:construct-mho}).}
\label{fig:construct-omega-G}
\end{figure}
\begin{proof}
  The vertices of the auxiliary graph $\mho(G)$ include
  $\overline{\oo}$, two copies of $\oo$, and $w$.  So the number of
  vertices is $2|\oo| + |\overline{\oo}| + 1 = |V(G)| + |\oo| + 1 \le
  2 n$.  In $\mho(G)$, there are two edges derived from every edge of
  $G[\oo]$ and one edge from every other edge of $G$.  All other edges
  are incident to $w$, and there are \occ\ of them.  Therefore, the
  number of edges is $|E(G)| + |E(G[\oo])| + |\occ| \le |E(G)| +
  |E(G[\oo])| + |\ecc| < 2 m$.  This concludes the proof of the first
  assertion.

  For the construction of $\mho(G)$, we use the procedure described in
  Fig.~\ref{fig:construct-omega-G} (some bookkeeping details are
  omitted).  Step~1 adds vertex sets $L$ and $R$ (step 1.1) as well as
  those edges induced by them (step 1.2.1), and finds $N(\oo)$ (step
  1.2.2).  Step~2 adds edges in $\ecc$ and $\ec$, and detect \occ\ and
  \oc.  Steps~3 and 4 add vertex $w$ and edges incident to it.  Step~5
  cleans $\oo$.  The dominating steps are 1 and 2, each of which
  checks every edge at most once, and hence the total time is $O(n +
  m)$.
\end{proof}

In an interval model, each vertex $v$ corresponds to a closed interval
$I_v = [\lp{v}, \rp{v}]$.  Here \lp{v} and \rp{v} are the left and
right \emph{endpoints} of $I_v$ respectively, and $\lp{v}< \rp{v}$.
We use unit-length circles for circular-arc models, where every point
has a positive value in $(0,1]$.  Each vertex $v$ corresponds to a
closed arc $A_v = [\cce{v}, \ce{v}]$.  Here \cce{v} and \ce{v} are
counterclockwise and clockwise endpoints of $A_v$ respectively; $0<
\cce{v}, \ce{v}\le 1$ and they are assumed to be distinct.  It is
worth noting that possibly $\ce{v}<\cce{v}$; such an arc necessarily
contains the point $1$.

The definition of $\mho(G)$ is motivated by the following observation;
see Fig.~\ref{fig:cag-interval}.  We remark that this lemma is
actually implied by Theorem~\ref{thm:negative-certificate}, and the
constructive proof presented below is to reveal the intuition.
\begin{lemma}\label{lem:nhcag-to-interval}
  If $G$ is a normal Helly circular-arc graph, then $\mho(G)$ is an
  interval graph.
\end{lemma}
\begin{figure}[h]
  \centering
  \subfloat[A normal and Helly circular-arc model of a graph $G$.]
  {\label{fig:xxx}
    \includegraphics{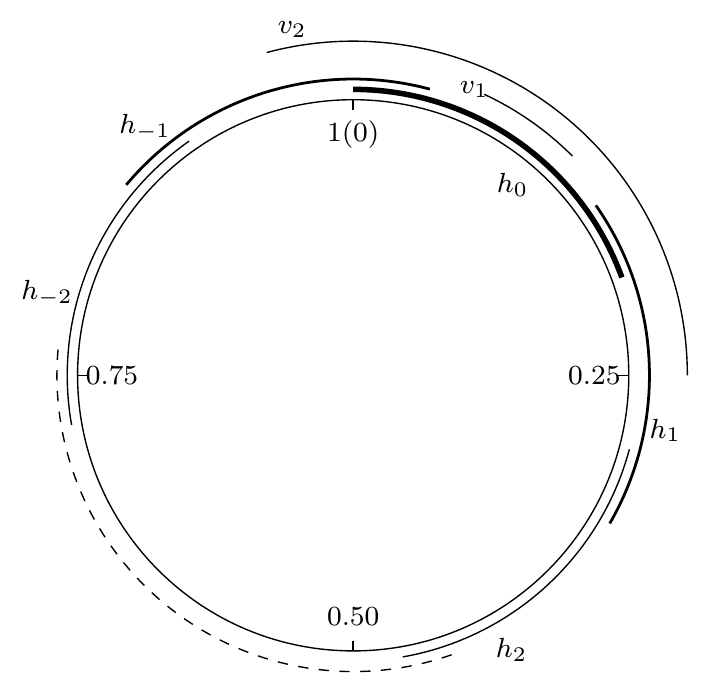}
  }

  \subfloat[{An interval model of $\mho(G)$ derived from
    (a)}.]{\label{fig:yyy}
    \includegraphics{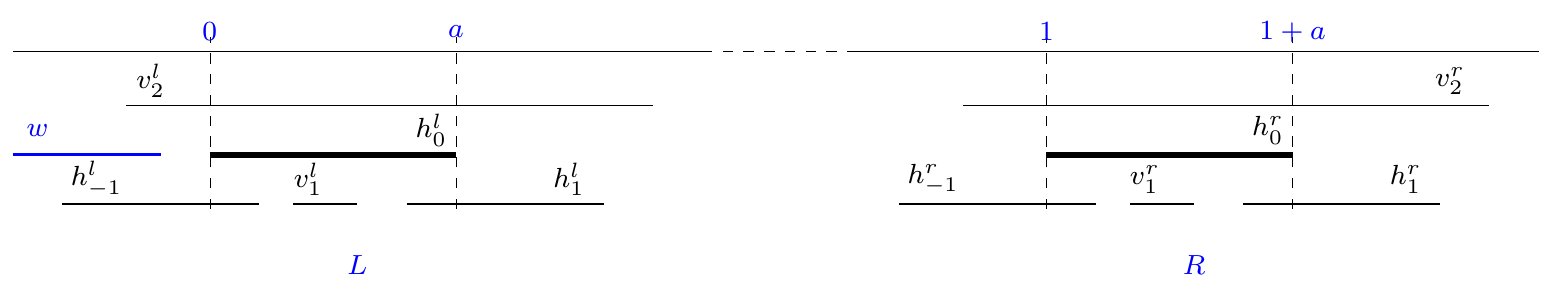}
  }
  \caption{Illustration for Lemma~\ref{lem:nhcag-to-interval}.  }
  \label{fig:cag-interval}
\end{figure}
\begin{proof}
  Let $\cal A$ be a normal Helly circular-arc model of $G$.  Then
  the union of arcs $\{A_v : v\in \oo\}$ does not cover the circle.
  Assume that $\cce{h_0} = 1$ and no other arc has an end at $1$.  An
  interval model of $\mho(G)$ can be obtained from $\cal A$ by setting
\begin{itemize}
\item $I_{v^l} := [\cce{v}, \ce{v}]$ and $I_{v^r} := [\cce{v} + 1,
  \ce{v} + 1]$ for $v\in \oo$ with $1\not\in A_v$;
\item $I_{v^l} := [\cce{v} - 1, \ce{v}]$ and $I_{v^r} := [\cce{v},
  \ce{v} + 1]$ for $v\in \oo$ with $1\in A_v$;
\item $I_{u} := [\cce{u}, \ce{u}]$ for $u\in \overline \oo$;  and
\item $I_{w} := [-1, \max_{x\in \overline \oo} \rp{x} - 1]$.
\end{itemize}
It is easy to use the definition to verify that this gives an interval
model of $\mho(G)$.
\end{proof}

Note that for any vertex $v\in \oo$, an induced \stpath{v^l}{v^r}
corresponds to a cycle whose arcs cover the entire circle.  The main
thrust of our algorithm will be a process that does the reversed
direction, which is nevertheless far more involved.

\begin{theorem}\label{lem:build-ca-model}
  If $\mho(G)$ is an interval graph, then we can in $O(n+m)$ time
  build a circular-arc model of $G$.
\end{theorem}
\begin{proof}
  We can in $O(n+m)$ time build an interval model $\cal I$ for
  $\mho(G)$.  By construction, ($w h_{-1}^l h_0^l h_1^l h_2 \cdots$
  $h_{-2} h_{-1}^r h_0^r h_1^r$) is an induced path of $\mho(G)$;
  without loss of generality, assume it goes ``from left to right'' in
  $\cal I$.  We may assume $\rp{w} = 0$ and $\max_{u\in
    \overline{\oo}}\rp{u} = 1$, while no other interval in $\cal I$
  has $0$ or $1$ as an endpoint.  Let $a = \rp{h_0^l}$.  We use $\cal
  I$ to construct a set of arcs for $V(G)$ as follows.  For each $u\in
  \overline{\oo}$, we set $A_u := [\lp{u}, \rp{u} ]$, which is a
  subset of the arc $(a,1]$.  For each $v\in \oo$, we set
  \begin{equation*}
    \label{eq:arcs}
    A_v := 
    \begin{cases}
      [\lp{v^r}, \rp{v^l} ] & \text{if } v\in\occ,
      \\
      [\lp{v^l}, \rp{v^l} ] & \text{otherwise}.
    \end{cases}
  \end{equation*}

  It remains to verify that the arcs obtained as such represent $G$,
  i.e., a pair of vertices $u, v$ of $G$ is adjacent if and only if
  $A_u$ and $A_v$ intersect.  This holds trivially when neither $u$
  nor $v$ is in $\oo$; hence we may assume without loss of generality
  that $v\in \oo$.  By construction, $a< \lp{u}<\rp{u}\le 1$ for every
  $u\in \overline{\oo}$.  Note that $v^l\sim w$ and $v^r\sim
  \overline{\oo}$ for every $v\in \occ$, which implies that $\lp{v^l}<
  0$ if and only if $\lp{v^r}< 1$ if and only if $v\in \occ$.

  Assume first that $u$ is also in $\oo$, then $u\sim v$ in $G$ if and
  only if $u^l\sim v^l$ in $\mho(G)$.  They are adjacent when both
  $u,v\in\occ$, and since $\lp{v^l}, \lp{u^l}< 0$, both $A_u$ and
  $A_v$ contains the point $1$ and thus intersect.  If neither $u$ nor
  $v$ is in $\occ$, then $\lp{v^l}, \lp{u^l}> 0$, and $u\sim v$ if and
  only if $A_u = [\lp{u^l}, \rp{u^l}]$ and $A_v = [\lp{v^l},
  \rp{v^l}]$ intersect.  Otherwise, assume, without loss of
  generality, that $\lp{v^l} < 0 < \lp{u^l}$, then $u\sim v$ in $G$ if
  and only if $0 < \lp{u^l} < \rp{v^l}$, which implies $A_u$ and $A_v$
  intersect (as both contain $[\lp{u^l}, \rp{v^l}]$).

  Assume now that $u$ is not in $\oo$, and then $u\sim v$ in $G$ if
  and only if either $u\sim v^l$ or $u\sim v^r$ in $\mho(G)$.  In the
  case $u\sim v^l$, we have $\lp{v^l} \le a < \lp{u}\le \rp{v^l}$;
  since both $A_u$ and $A_v$ contain $[\lp{u}, \rp{v^l}]$, which is
  nonempty, they intersect.  In the case $u\sim v^r$, we have
  $\lp{v^r}< \rp{u} \le 1$; since both $A_u$ and $A_v$ contain
  $[\lp{v^r}, \rp{u}]$, they intersect.  Otherwise, $u\not\sim v$ in
  $G$ and $\lp{v^l} < \rp{v^l} < \lp{u} < \rp{u} < \lp{v^r} <
  \rp{v^r}$, then $A_u$ and $A_v$ are disjoint.
\end{proof}

We are now ready to present the recognition algorithm in
Fig.~\ref{fig:main-alg}, and prove
Theorem~\ref{thm:certifying-algorithm}.  Recall that Lin et
al.~\cite{lin-13-nhcag-and-subclasses} have given a linear-time
algorithm for verifying whether a circular-arc model is normal and
Helly.  In fact, given a circular-arc model, we can in linear time
find the minimum number of arcs that cover the circle
\cite{cao-14-almost-interval-recognition}.
\begin{figure*}[h]
\setbox4=\vbox{\hsize28pc \noindent\strut
\begin{quote}
  \vspace*{-5mm} \footnotesize

  Algorithm {\bf nhcag}($G$)
  \\
  Input: a graph $G$.
  \\
  Output: either a normal Helly circular-arc model of $G$, or a
  forbidden induced subgraph of $G$.
  \\
  
  1 \hspace*{2ex} \parbox[t]{0.90\linewidth}{ test the chordality of
    $G$ and find a hole $H$ if not;
    \\
    {\bf if} $G$ is chordal {\bf then} verify whether $G$ is an
    interval graph or not;}
  \\
  2 \hspace*{2ex} construct the auxiliary grpah $\mho(G)$;
  \\
  3 \hspace*{2ex} {\bf if} $\mho(G)$ is not an interval graph {\bf
    then}
  \\
  \hspace*{7ex} {\bf call} Theorem~\ref{thm:negative-certificate} to find
  a forbidden induced subgraph;
  \\
  4 \hspace*{2ex} {\bf call} Theorem~\ref{lem:build-ca-model} to build a
  circular-arc $\cal A$ model of $G$;
  \\
  5 \hspace*{2ex} verify whether $\cal A$ is normal and Helly.
\end{quote} \vspace*{-6mm} \strut} $$\boxit{\box4}$$
\vspace*{-9mm}
\caption{The recognition algorithm for \nhcag s.}
\label{fig:main-alg}
\end{figure*}

\begin{proof}[Proof of Theorem~\ref{thm:certifying-algorithm}]
  We use the algorithm {nhcag} presented in
  Fig.~\ref{fig:main-alg}.  Step 1 is clear.  Steps 2 to 4 follow from
  Lemma~\ref{lem:construct-mho}, Theorem~\ref{thm:negative-certificate},
  and Lemma~\ref{lem:build-ca-model}, respectively.  If the model $\cal
  A$ built in step 4 is not normal and Helly, then we can in linear
  time find such a set of two or three arcs whose union covers the
  circle.  Their corresponding vertices satisfy
  Lemma~\ref{lem:non-helly}, and this concludes the proof.
\end{proof}
It is worth noting that if we are after a recognition algorithm (with
positive certificate only), then we can simply return ``NO'' if the
hypothesis of step 3 is true (see Lemma~\ref{lem:nhcag-to-interval})
and the algorithm is already complete.

\section{Proof of Theorem~\ref{thm:negative-certificate}}
Recall that Theorem~\ref{thm:negative-certificate} is only called in
step 3 of algorithm nhcag; the graph is then not choral and we have a
hole $H$.  In principle, we can pick any vertex as $h_0$.  But for the
convenience of presentation, we require it satisfies some additional
conditions.  If some vertex $v$ is adjacent to four or more vertices
in $H$, i.e., $\head{v}-\tail{v} > 2$, then $v\not\in H$.  We can thus
use ($h_{\tail{v}} v h_{\head{v}}$) as a short cut for the sub-path
induced by $N_H[v]$, thereby yielding a strictly shorter hole.  This
condition, that $h_0$ cannot be bypassed as such, is formally stated
in the following lemma, and a procedure for finding such a hole is
given in Fig.~\ref{fig:compress-hole}.
\begin{lemma}\label{lem:hole-conditions}
  We can in $O(n+m)$ time find either a \badgraph, or a hole $H$ such
  that $\{h_{-1}, h_0, h_1\}\subseteq N_H[v]$ for some $v$ if and only
  if $N_H[v] = \{h_{-1}, h_0, h_1\}$.
\end{lemma}
\begin{figure}[h]
\setbox4=\vbox{\hsize28pc \noindent\strut
\begin{quote}
  \vspace*{-5mm} \footnotesize

  {\sc input}: a graph $G$ and a hole $H$ of $G$.
  \\
  {\sc output}: a hole satisfying conditions of
  Lemma~\ref{lem:hole-conditions} or a \badgraph.
  \\[1ex]
  0 \hspace*{2ex} $h = h_0$; $a = -1$; $b = 1$;
  \\
  1 \hspace*{2ex} {\bf for each} $v\in V(G)\setminus H$ {\bf do}
  \\
  1.1 \hspace*{4ex} compute \tail{v} and \head{v} in $H$;
  \\
  1.2 \hspace*{4ex} {\bf if} \big($\tail{v}< a$ and $\head{v} \ge
  b$\big) or \big($\tail{v}= a$ and $\head{v} >b$\big) {\bf then}
  \\
  \hspace*{10ex} $h = v$; $a = \tail{v}$; $b = \head{v}$;
  \comment{$\tail{v}< 0 < \head{v}$.}
  \\
  2 \hspace*{2ex} {\bf return} ($h h_{b} h_{b + 1} \cdots h_{a} h$)
  where $h$ is the new $h_0$.
\end{quote} \vspace*{-3mm} \strut} $$\boxit{\box4}$$
\vspace*{-7mm}
\caption{Procedure for finding the hole for Lemma~\ref{lem:hole-conditions}.}
\label{fig:compress-hole}
\end{figure}
\begin{proof}
  We apply the procedure given in Fig.~\ref{fig:compress-hole}.
  Step~1 greedily searches for an inclusion-wise maximal $N_H[v]$
  satisfying $\tail{v}\le -1$ and $1 \le \head{v}$.  Initially,
  $a=\tail{h_0}=-1$ and $b=\head{h_0}=1$.  Each iteration of step~1
  checks an unexplored vertex $v$ in $V(G)\setminus H$.  If either
  condition of step 1.2 is satisfied, then $N[v]$ properly contains
  $\{h_a,h_{a+1},\dots,h_b\}$, and $a$ and $b$ are updated to be
  \tail{v} and \head{v} respectively.  Note that the values of $a$ and
  $b$ are non-increasing and nondecreasing respectively.  Thus, no
  previously explored vertex is adjacent to all of
  $\{h_a,h_{a+1},\dots,h_b\}$.  After step 1, all vertices have been
  explored, and the hole ($h h_{b} h_{b + 1} \cdots h_{a} h$)
  satisfies the claimed condition.

  What dominates the procedure is finding \tail{v} and \head{v} for
  all vertices (step~1.1).  It takes $O(d(v))$ time for each vertex
  $v$ and $O(n+m)$ time in total.
\end{proof}
This linear-time procedure can be called before step 2 of algorithm
{nhcag}, and it does not impact the asymptotic time complexity of
the algorithm, which remains linear.  Henceforth we may assume that
$H$ satisfies the condition of Lemma~\ref{lem:hole-conditions}.  During
the construction of $\mho(G)$, we have checked $N_H[v]$ for every
vertex $v$, and Lemma~\ref{lem:non-consecutive} was called if it
applies.  Therefore, for the proof of
Theorem~\ref{thm:negative-certificate} in this section, we may assume
that $N_H[v]$ always induces a proper sub-path of $H$.

Each vertex $x$ of $\mho(G)$ different from $w$ is uniquely defined by
a vertex of $G$, which is denoted by $\og{x}$.  We say that $x$ is
\emph{derived from} $\og{x}$.   For example, $\og{v^l} = \og{v^r} = v$
for $v\in \oo$.  By abuse of notation, we will use the same letter for
a vertex $u\in \overline{\oo}$ of $G$ and the unique vertex of
$\mho(G)$ derived from $u$; its meaning is always clear from the
context.  Therefore, $\og{u} = u$ for $u\in \overline{\oo}$, and in
particular, $\og{h_i} = h_i$ for $i = 2, \dots, |H|-2$.  We can mark
$\og{x}$ for each vertex of $\mho(G)$ during its construction.  The
function $\phi$ is also generalized to a set $U$ of vertices that does
not contain $w$, i.e., $\og{U} = \{\og{v}: v\in U\}$.  We point out
that possibly $|\og{U}|\ne |U|$.

By construction of $\mho(G)$, if a pair of vertices $x$ and $y$
(different from $w$) is adjacent in $\mho(G)$, then $\og{x}$ and
$\og{y}$ must be adjacent in $G$ as well.  The converse is not
necessarily true, e.g., for any vertex $v\in\oc$ and edge $u v \in
\ec$, we have $u\not\sim v^r$, and for any pair of adjacent vertices
$u,v\in \oo$, we have $u^l\not\sim v^r$ and $u^r\not\sim v^l$.  We say
that a pair of vertices $x,y$ of $\mho(G)$ is a \emph{bad pair} if
$\og{x}\sim \og{y}$ in $G$ but ${x}\not\sim {y}$ in $\mho(G)$.  By
definition, $w$ does not participate in any bad pair, and at least one
vertex of a bad pair is in $L\cup R$.  Note that any induced path of
length $d$ between a bad pair $x,y$ with $x=v^l$ or $v^r$ can be
extended to a \stpath{v^l}{v^r} with length $d+1$.

We have seen that if $G$ is a \nhcag, then for any $v\in \oo$, the
distance between $v^l$ and $v^r$ is at least $4$.  We now see what
happens when this necessary condition is not satisfied by $\mho(G)$.
By definition of $\mho(G)$, there is no edge between $L$ and $R$; for
any $v\in \oo$, there is no vertex adjacent to both $v^l$ and $v^r$.
In other words, for every $v\in \oo$, the distance between $v^l$ and
$v^r$ is at least $3$.  The following observation can be derived from
Lemmas.~\ref{lem:non-consecutive} and \ref{lem:non-consecutive-2}.
\begin{lemma}\label{lem:3-cover}
  Given a \stpath{v^l}{v^r} $P$ of length $3$ for some $v\in \oo$, we
  can in $O(n + m)$ time find a \badgraph.
\end{lemma}
\begin{proof}
  Let $P = (v^l x y v^r)$.  Note that $P$ must be a shortest
  \stpath{v^l}{v^r}, and $w\not\in P$.  The inner vertices $x$ and $y$
  cannot be both in $L\cup R$; without loss of generality, let $x\in
  \overline{\oo}$.  Assume first that $y\in \overline{\oo}$ as well,
  i.e., $v x\in \ec$ and $v y\in \ecc$.  By definition, $v\in
  \oc\cap\occ$, and then $v$ is adjacent to both $h_{-1}$ and $h_1$.
  If follows from Lemma~\ref{lem:hole-conditions} that $N_H[v] =
  \{h_{-1}, h_0, h_1\}$, and then $x\sim h_1$ and $y\sim h_{-1}$.  If
  $x\sim h_{-1}$, i.e., $\head{x} = |H| - 1$, then we call
  Lemma~\ref{lem:non-consecutive-2} with $v$ and $x$.  If $\head{x} <
  \tail{y}$, then we call Lemma~\ref{lem:non-consecutive} with $x$ and
  $y$.  In the remaining case, $\tail{y}\le \head{x}< |H| - 1$, and
  ($v x h_{\head{x}} \cdots h_{-1} v$) is a hole of $G$; this hole is
  completely adjacent to $y$, and thus we find a wheel.

  Now assume that, without loss of generality, $y = u^r \in R$.  If
  $\head{v} \ge \tail{y}$, then we call
  Lemma~\ref{lem:non-consecutive-2} with $v$ and $y$.  Otherwise, ($v
  h_{\head{v}} \cdots h_{\tail{y}} u v$) is a hole of $G$; this hole
  is completely adjacent to $x$, and thus we find a wheel.
\end{proof}

If $G$ is a \nhcag, then in a circular-arc model of $G$, all arcs for
\occ\ and \oc\ contain $\cce{h_0}$ and $\ce{h_0}$ respectively.  Thus,
both \occ\ and \oc\ induce cliques.  This observation is complemented
by the following lemma.
\begin{lemma}\label{lem:O}
  Given a pair of nonadjacent vertices $u,x\in \occ$ (or \oc), we can
  in $O(n+m)$ time find a \badgraph\ of $G$.
\end{lemma}
\begin{proof}
  By definition, we can find edges $u v,x y\in \ecc$.  We have three
  (possibly intersecting) induced paths $h_0 h_1 h_2$, $h_0 u v$, and
  $h_0 x y$.  If both $u$ and $x$ are adjacent to $h_1$, then we
  return ($u h_{-1} x h_1 u$)$+h_0$ as a wheel.  Hence we may assume
  $x\not\sim h_1$.

  If $u\sim h_1$, then by Lemma~\ref{lem:hole-conditions}, $N_H[u] =
  \{h_{-1}, h_0, h_1\}$.  We consider the subgraph induced by the set
  of distinct vertices $\{h_0, h_1, h_2, u, v, x\}$.  If $v$ is
  adjacent to $h_0$ or $h_1$, then we can call
  Lemma~\ref{lem:non-helly} with $u$ and $v$.  By assumption, $h_0,
  h_1$, and $u$ make a triangle; $x$ is adjacent to neither $u$ nor
  $h_1$; and $h_2$ is adjacent to neither $h_0$ nor $u$.  Therefore,
  the only uncertain adjacencies in this subgraph are between $v, x$,
  and $h_2$.  The subgraph is thus isomorphic to
  \begin{inparaitem}
  \item[(1)] FIS-1 if there are two edges among $v, x$, and $h_2$;
  \item[(2)] $\overline{C_6}$ if $v, x$, and $h_2$ are pairwise
    adjacent; or
  \item[(3)] net if $v, x$, and $h_2$ are pairwise nonadjacent.
  \end{inparaitem}
  In the remaining cases there is precisely one edge among $v, x$, and
  $h_2$, then we can return a $C^*$, e.g., ($v x h_0 u v$) and $h_2$
  when the edge is $vx$.

  Assume now that $u, x$, and $h_1$ are pairwise nonadjacent.  We
  consider the subgraph induced by $\{h_0, h_1, h_2, u, v, x, y\}$,
  where the only uncertain relations are between $v, y$, and $h_2$.
  The subgraph is thus isomorphic to
  \begin{inparaenum}[(1)]
  \item $K_{2,3}$ if all of them are identical; or
  \item twin-$C_5$ if two of them are identical, and adjacent to the
    other.
  \end{inparaenum}
  If two of them are identical, and nonadjacent to the other, then the
  subgraph contains a $C^*$, e.g., ($v u h_0 x v$) and $h_2$ when
  $v=y$.  In the remaining cases, all of $v, y$, and $h_2$ are
  distinct, and then the subgraph
  \begin{inparaenum}[(1)]
  \item  is isomorphic to long claw if they are pairwise nonadjacent;
  \item contains net $\{h_1, h_2, u, v, x, y\}$ if they are pairwise
    adjacent; or
  \item  is isomorphic to FIS-2 if there are two edges among them.
  \end{inparaenum}
  If there is one edge among them, then the subgraph contains a $C^*$,
  e.g., ($v u h_0 x y v$) and $h_2$ when the edge is $v y$.

  A symmetrical argument applies to \oc.  Edges $u v$ and $x y$ can be
  found in $O(n)$ time, and only a small constant number of
  adjacencies are checked; it thus takes $O(n+m)$ time in total.
\end{proof}

It can be checked in linear time whether \occ\ and \oc\ induce
cliques.  When it is not, a pair of nonadjacent vertices can be found
in the same time.  By Lemma~\ref{lem:O}, we may assume hereafter that
\occ\ and \oc\ induce cliques.  We say that a vertex $v$ is {\em
  simplicial} if $N[v]$ induces a clique.  Recall that $N(w)\subseteq
\occ$; as a result, $w$ is simplicial and participates in no holes.
\begin{proposition}\label{lem:non-bypass}
  Given an \stpath{h^l_0}{h^r_0} $P$ that is nonadjacent to $h_i$ for
  some $1<i<|H| - 1$, we can in $O(n + m)$ time find a \badgraph.
\end{proposition}
\begin{proof}
  Inside $P$ there must be a sub-path $P'$ whose ends $x,y$ are in $L$
  and $R$ respectively, and whose inner vertices are all in
  $\overline\oo$.  Let $x'$ and $y'$ be the neighbors of $x$ and $y$
  in $P'$ respectively; note that they are both in $\overline\oo$.  If
  $N_H[\og{x}]$ and $N_H[x']$ are disjoint, then we call
  Lemma~\ref{lem:non-consecutive-2}.  Otherwise, by assumption, we
  have $0 < \tail{x'} \le \head{x'} < i$; likewise, we may assume
  $i<\tail{y'}\le \head{y'}< |H|$.  Starting from $x'$, we traverse
  $P'$ till the first pair of consecutive vertices $u,v$ in $P$ such
  that $N_H[u]$ and $N_H[v]$ are disjoint: note that such a pair must
  exist because no vertex between $x'$ and $y'$ is adjacent to $h_0$
  or $h_i$.  Then we call Lemma~\ref{lem:non-consecutive-2}.
\end{proof}

We are now ready to prove Theorem~\ref{thm:negative-certificate}, which
is separated into three statements, the first of which considers the
case when $\mho(G)$ is not chordal.
\begin{lemma}\label{lem:hole}
  Given a hole $C$ of $\mho(G)$, we can in $O(n+m)$ time find a
  \badgraph.
\end{lemma}
\begin{proof}
  Let us first take care of some trivial cases.  If $C$ is contained
  in $L$ or $R$ or $\overline{\oo}$, then by construction, \og{C} is a
  hole of $G$.  This hole is either nonadjacent or completely adjacent
  to $h_0$ in $G$, whereupon we can return \og{C} and $h_0$ as a $C^*$
  or wheel respectively.  Since $L$ and $R$ are nonadjacent, one of
  the cases above must hold if $C$ is disjoint from $\overline{\oo}$.
  Henceforth we may assume that $C$ intersects $\overline{\oo}$ and,
  without loss of generality, $L$; it might intersect $R$ as well, but
  this fact is irrelevant in the following argument.  Then we can find
  an edge $x_1 x_2$ of $C$ such that $x_1\in L$ and $x_2\in
  \overline{\oo}$, i.e., $x_1 x_2\in \ec$.

  Let $a := {\head{\og{x_1}}}$.  Assume first that $x_2 = h_a$; then
  we must have $a > 1$.  Let $x_3$ and $x_4$ be the next two vertices
  of $C$.  Note that $x_3\not\in L$, i.e., $x_3\not\sim h^l_0$;
  otherwise $x_1\sim x_3$, which is impossible.  If $x_3\sim h_{a -
    2}$ (or $h^l_{a - 2}$ when $a = 3$), then
  $\og{\{x_1,x_2,x_3\}}\cup \{h_{a - 2}\}$ induces a hole of $G$, and
  we can return it and $h_{a-1}$ as a wheel.  Note that $x_4 \not\sim
  h_{a}$ as they are non-consecutive vertices of the hole $C$.  We now
  argue that $\head{\og{x_4}}< a$.  Suppose for contradiction,
  $\tail{\og{x_4}} > a$.  We can extend the \stpath{x_3}{x_1} $P$ in
  $C$ that avoids $x_2$ to a \stpath{h^l_0}{h^r_0} avoiding the
  neighborhood of $h_a$, which allows us to call
  Proposition~\ref{lem:non-bypass}.  We can call
  Lemma~\ref{lem:non-consecutive-2} with $x_3$ and $x_4$ if
  $\tail{\og{x_3}} = a$.  In the remaining case, $\tail{\og{x_3}} = a
  - 1$.  Let $x$ be the first vertex in $P$ that is adjacent to
  $h_{a-2}$ (or $h^l_{a-2}$ if $a\le 3$); its existence is clear as
  $x_1$ satisfies this condition.  Then $\og{\{x_3, \dots, x, h_{a-2},
    x_1, x_2\}}$ induces a hole of $G$, and we can return it and
  $h_{a-1}$ as a wheel.

  Assume now that $h_a$ is not in $C$.  Denote by $P$ the
  \stpath{x_2}{x_1} obtained from $C$ by deleting the edge $x_1 x_2$.
  Let $x$ be the first neighbor of $h_{a+1}$ in $P$, and let $y$ be
  either the first neighbor of $h_{a-1}$ in the \stpath{x}{x_1} or the
  other neighbor of $x_1$ in $C$.  It is easy to verify that
  $\og{\{x_1, \cdots, x, \cdots y, x_2\}}$ induces a hole of $G$,
  which is completely adjacent to $h_a$, i.e., we have a wheel.
\end{proof}

In the rest $\mho(G)$ will be chordal, and thus we have a chordal
minimally non-interval subgraph $F$ of $\mho(G)$.  This subgraph is
isomorphic to some graph in Fig.~\ref{fig:at}, on which we use the
following notation.  It is immediate from Fig.~\ref{fig:at} that each
of them contains precisely three simplicial vertices (squared
vertices), which are called \emph{terminals}, and others (round
vertices) are \emph{non-terminal vertices}.  In a long claw or \dag,
for each $i=1,2,3$, terminal $t_i$ has a unique neighbor, denoted by
$u_i$.

Since the diameter and maximum clique of $F$ is at most four, all bad
pairs in it can be found easily.
\begin{proposition}\label{lem:find-bad-pair}
  Given a subgraph $F$ of $\mho(G)$ in Fig.~\ref{fig:at}, we can in
  $O(n+m)$ time find either all bad pairs in $F$ or a \badgraph.
\end{proposition}
\begin{proof}
  A long claw or whipping top has only seven vertices, and thus will
  not concern us.  A bad pair is either between $L$ and $R$, or
  between $\{v^l: v\in \occ\}\cup \{v^r: v\in \oc\}$ and $\overline
  \oo$.  Let ${L_{\text{c}}} = \{v^l: v\in \oc\}$ and ${R_{\text{cc}}}
  = \{v^r: v\in \occ\}$; both induce cliques.  Since a clique of $F$
  contains at most $4$ vertices, it contains at most $8$ vertices of
  ${L_{\text{c}}}$ and ${R_{\text{cc}}}$.  Bad pairs intersecting them
  can thus be found in linear time.  We now consider other bad pairs,
  which must be between $L\setminus {L_{\text{c}}}$ and $R\setminus
  {R_{\text{cc}}}$.  By construction, there is no edge between
  $L\setminus {L_{\text{c}}}$ and $R\cup \overline\oo$; there is no
  edge between $R\setminus {R_{\text{cc}}}$ and $L\cup \overline\oo$.
  Therefore, the shortest distance between $L\setminus {L_{\text{c}}}$
  and $R\setminus {R_{\text{cc}}}$ is at least $4$.  There exists only
  one pair of distance $4$ vertices in a $\dag$, and no such a pair in
  a $\ddag$.
\end{proof}

\begin{lemma}\label{lem:at-no-w}
  Given a subgraph $F$ of $\mho(G)$ in Fig.~\ref{fig:at} that does not
  contain $w$, we can in $O(n+m)$ time find a \badgraph.
\end{lemma}
\begin{figure}[h]
  \centering
  \footnotesize
  \subfloat[labeled long claw]{\label{fig:long-claw}
    \includegraphics{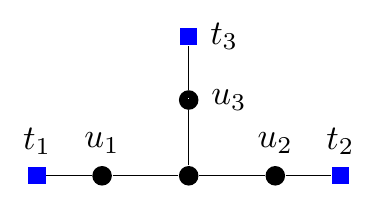} 
  }
  $\qquad$
  \subfloat[]{\label{fig:long-claw-5}
    \includegraphics{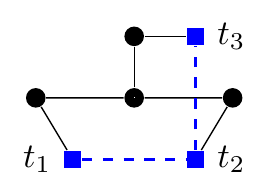} 
  }
  $\qquad$
  \subfloat[]{\label{fig:long-claw-6}
    \includegraphics{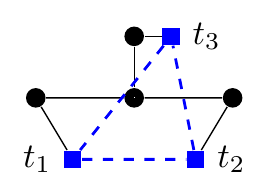} 
  }

  \subfloat[]{\label{fig:long-claw-1}
    \includegraphics{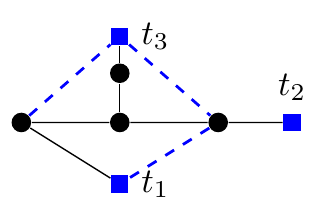} 
  }
  $\quad$
  \subfloat[]{\label{fig:long-claw-2}
    \includegraphics{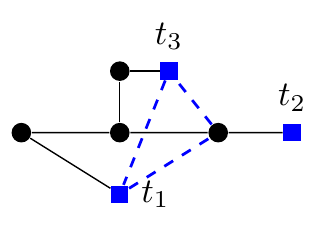} 
  }
  $\quad$
  \subfloat[]{\label{fig:long-claw-3}
    \includegraphics{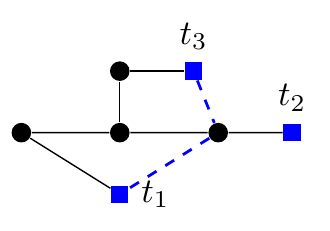} 
  }
  $\quad$
  \subfloat[]{\label{fig:long-claw-4}
    \includegraphics{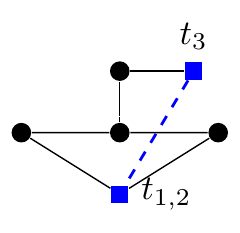} 
  }
  \caption{Illustrations for Lemma~\ref{lem:at-no-w} (blue dashed
    edges are in $G$ only).}
\label{fig:negative-certificate-1}
\end{figure}
\begin{proof}
  We first call Proposition~\ref{lem:find-bad-pair} to find all bad
  pairs in $F$.  If $F$ has no bad pair, then we return the subgraph
  of $G$ induced by \og{F}, which is isomorphic to $F$.  Let $x,y$ be
  a bad pair with the minimum distance in $F$; we may assume that it
  is $3$ or $4$, as otherwise we can call Lemma~\ref{lem:3-cover}.
  Noting that the distance between a pair of non-terminal vertices is
  at most $2$, we may assume that without loss of generality, $x$ is a
  terminal of $F$.  We break the argument based on the type of $F$.

  {\em Long claw.}  We may assume that $x = t_1$ and $y\in
  \{u_2,t_2\}$; other situations are symmetrical.  Let $P$ be the
  unique \stpath{x}{y} in $F$.  If \og{t_3} is nonadjacent to \og{P},
  then we return \og{P} and \og{t_3} as a $C^*$; we are thus focused
  on the adjacency between \og{t_3} and \og{P}
  (Fig.~\ref{fig:negative-certificate-1}).  If $y = t_2$, then by the
  selection of $x,y$ (they have the minimum distance among all bad
  pairs), $\og{t_3}$ can be only adjacent to $\og{t_1}$ and/or
  $\og{t_2}$.  We return either $\og{F}$ as an FIS-2
  (Fig.~\ref{fig:long-claw-5}), or $\og{\{t_1,t_2,t_3,u_1,u_2,u_3\}}$
  as a net (Fig.~\ref{fig:long-claw-6}).  In the remaining cases, $y =
  u_2$, and $\og{t_3}$ can only be adjacent to \og{u_1}, \og{u_2},
  and/or \og{t_1}.  We point out that possibly $\og{t_2} = \og{t_1}$,
  which is irrelevant as $\og{t_2}$ will not be used below.  If
  $\og{t_3}$ is adjacent to both $\og{u_1}$ and $\og{u_2}$ in $G$,
  then we get a $K_{2,3}$ (Fig.~\ref{fig:long-claw-1}).  Note that
  this is the only case when $\og{t_1} = \og{t_3}$.  If $\og{t_3}$ is
  adjacent to both $\og{t_1}$ and $\og{u_2}$ in $G$, then we get an
  FIS-1 (Fig.~\ref{fig:long-claw-2}).  If $\og{t_3}$ is adjacent to
  only $\og{u_2}$ or only $\og{t_1}$ in $G$, then we get a domino
  (Fig.~\ref{fig:long-claw-3}) or twin-$C_5$
  (Fig.~\ref{fig:long-claw-4}), respectively.  The situation that
  $\og{t_3}$ is adjacent to \og{u_1} but not \og{u_2} is similar as
  above.

  {\em Whipping top.} The diameter is $3$, and this distance is
  attained only by $\{t_1,t_3\}$ or $\{t_2, t_3\}$.  If both are bad
  pairs, then we have a domino.  If $\{t_1,t_3\}$ is the only bad
  pair, then $\og{F\setminus N[t_2]}$ induces a hole of $G$, and it is
  nonadjacent to \og{t_2}; we get a $C^*$.  A symmetrical argument
  applies if $\{t_2,t_3\}$ is the only bad pair.

  {\em \dag.}  Consider first that $x = t_1$ and $y = t_3$, and let $P
  = (t_1 u_1 u_3 t_3)$.  If \og{t_2} is nonadjacent to the hole
  induced by \og{P}, then we return \og{P} and \og{t_2} as a $C^*$.
  If \og{t_2} is adjacent to \og{t_3} or \og{u_1}, then we get a
  domino.  If \og{t_2} is adjacent to \og{t_1}, then we get a
  twin-$C_5$.  If \og{t_2} is adjacent to \og{t_1} and precisely one
  of $\{\og{t_3}, \og{u_1}\}$, then we get an FIS-1.  If \og{t_2} is
  adjacent to both \og{t_3} and \og{u_1}, then we get a $K_{2,3}$;
  here the adjacency between \og{t_2} and \og{t_1} is immaterial.  A
  symmetric argument applies when $\{t_2, t_3\}$ is a bad pair.  In
  the remaining case, neither \og{t_1} nor \og{t_2} is adjacent to
  \og{t_3}.  Therefore, a bad pair must be in the path $F - N[t_3]$,
  which is nonadjacent to \og{t_3}, then we get a $C^*$.

  {\em \ddag.} The only pair of vertices of distance $3$ is $\{t_1,
  t_2\}$.  Let $P$ be the \stpath{t_1}{t_2} in $F - N[t_3]$.  Since
  $\og{t_3}$ cannot be adjacent to any vertex in $\og{P}$, we can
  return \og{P} and \og{t_3} as a $C^*$.
\end{proof}

\begin{lemma}\label{lem:at-with-w}
  Given a subgraph $F$ of $\mho(G)$ in Fig.~\ref{fig:at} that contains
  $w$, we can in $O(n+m)$ time find a \badgraph.
\end{lemma}
\begin{figure}[h]
\setbox4=\vbox{\hsize28pc \noindent\strut
\begin{quote}
  \vspace*{-5mm} \footnotesize

  \hspace*{2ex} $\setminus\!\!\setminus$ Note that $0\le
  \head{\og{x_1}}, \head{\og{x_2}} \le 1$.
  \\
  1 \hspace*{2ex} {\bf if} ${\head{\og{x_1}} = 1}$ and $y_1\sim h_2$
  {\bf then}
  \\
  \hspace*{7ex} {\bf call} Lemma~\ref{lem:non-consecutive-2} with ($y_1
  \og{x_1} h_{1} h_2 y_1$) and $\{\og{x_2}, y_2\}$;
  \\
  \textcolor{white}{1} \hspace*{2ex} {\bf if} ${\head{\og{x_1}} = 0}$
  and $y_1\sim h_1$ {\bf then}
  \\
  \hspace*{7ex} {\bf call} Lemma~\ref{lem:non-consecutive-2} with ($y_1
  \og{x_1} h_{0} h_1 y_1$) and $\{\og{x_2}, y_2\}$;
  \\
  \textcolor{white}{1} \hspace*{2ex} {\bf if} $y_2\sim
  h_{\head{\og{x_2}} + 1}$ {\bf then} {\sf symmetric as above};
  \\
  2 \hspace*{2ex} {\bf if} $\head{\og{x_1}} = \head{\og{x_2}}$ {\bf
    then}
  \\
  \hspace*{6ex} {\bf return} $\{y_1, \og{x_1}, y_2, \og{x_2},
  h_{\head{\og{x_2}}}, h_{\head{\og{x_2}}+1}\}$ as a \dag;
  \comment{Fig.~\ref{fig:hole-1}.}
  \\
  \hspace*{2ex} $\setminus\!\!\setminus$ assume from now that
  $\head{\og{x_1}} = 1$ and $\head{\og{x_2}} = 0$.
  \\
  3 \hspace*{2ex} {\bf if} ${\og{x_2}} \sim h_2$ {\bf then return}
  ($\og{x_2} h_0 h_1 h_2 \og{x_2}$) and $y_1$ as a $C^*$;
  \\
  4 \hspace*{2ex} {\bf if} $y_{2}\not\sim h_{-1}$ {\bf then return}
  $\{y_1, h_{-1}, \og{x_1}, y_2, \og{x_2}, h_0, h_1\}$ as a \ddag;
  \comment{Fig.~\ref{fig:hole-3}.}
  \\
  \textcolor{white}{4} \hspace*{2ex} {\bf if} $y_{2}\sim h_{-1}$ {\bf
    then return} $\{y_1, h_{-1}, y_2, \og{x_2}, h_0, h_1\}$ as a
  \dag. \comment{Fig.~\ref{fig:hole-4}.}

\end{quote} \vspace*{-3mm} \strut} $$\boxit{\box4}$$
\vspace*{-7mm}
\caption{Procedure for Lemma~\ref{lem:at-with-w}}
\label{fig:at-with-w}
\end{figure}
\begin{proof}
  Since $w$ is simplicial, it has at most $2$ neighbors in $F$.  If
  $w$ has a unique neighbor in $F$, then we can use a similar argument
  as Lemma~\ref{lem:at-no-w}.  Now let ${x_1}, x_2$ be the two
  neighbors of $w$ in $F$.  If there exists some vertex $u\in
  \overline{\oo}$ adjacent to both {\og{x_1}} and \og{x_2} in $G$,
  which can be found in linear time, then we can use it to replace $w$.
  Hence we assume there exists no such vertex.  By assumption, we can
  find two distinct vertices $y_1, y_2\in \overline{\oo}$ such that
  $\og{x_1} y_1, \og{x_2} y_2\in \ecc$; note that $\og{x_1}\not\sim
  y_2$ and $\og{x_2}\not\sim y_1$ in $G$.  As a result, $y_1$ and
  $y_2$ are nonadjacent; otherwise, $\{y_1, y_2\}$ and the
  counterparts of $\{{x_1}, x_2\}$ in $R$ induce a hole of $\mho(G)$,
  which contradicts the assumption that $\mho(G)$ is chordal.  We then
  apply the procedure described in Fig.~\ref{fig:at-with-w}.

  We now verify the correctness of the procedure.  Since each
  step---either directly or by calling a previously verified
  lemma---returns a \badgraph\ of $G$, all conditions of previous
  steps are assumed to not hold in a later step.  By
  Lemma~\ref{lem:hole-conditions}, $\head{\og{x_1}}$ and
  $\head{\og{x_2}}$ are either $0$ or $1$.  Step 1 considers the case
  where $y_1\sim h_{\head{\og{x_1}} + 1}$.  By
  Lemma~\ref{lem:non-helly}, $y_1\not\sim h_{\head{\og{x_1}}}$.  Thus,
  ($y_1 \og{x_1} h_{1} h_2 y_1$) or ($y_1 \og{x_1} h_{0} h_1 y_1$) is
  a hole of $G$, depending on ${\head{\og{x_1}}}$ is $0$ or $1$.  In
  the case ($y_1 \og{x_1} h_{1} h_2 y_1$), only $\og{x_1}$ and $h_1$
  can be adjacent to $\og{x_2}$; they are nonadjacent to $y_2$.
  Likewise, in the case ($y_1 \og{x_1} h_{0} h_1 y_1$), vertices
  $\og{x_1}$ and $h_0$ are adjacent to $\og{x_2}$ but not $y_2$, while
  $h_1$ can be adjacent to only one of $\og{x_2}$ and $y_2$.  Thus, we
  can call Lemma~\ref{lem:non-consecutive-2}.  A symmetric argument
  applies when $y_2\sim h_{\head{\og{x_2}} + 1}$.  Now that the
  conditions of step~1 do not hold true, step~2 is clear from
  assumption.  Henceforth we may assume without loss of generality
  that $\head{\og{x_1}} =1$ and $\head{\og{x_2}}= 0$.  Consequently,
  $\head{{y_1}} = |H| - 1$ (Lemma~\ref{lem:non-consecutive-2}).
  Because we assume that the condition of step~1 does not hold,
  $y_1\not\sim h_2$; this justifies step~3.  Step~4 is clear as $y_1$
  is always adjacent to $h_{-1}$.
  \end{proof}
\begin{figure}[h!]
  \centering
  \subfloat[]{\label{fig:hole-1}
    \includegraphics[scale=.97]{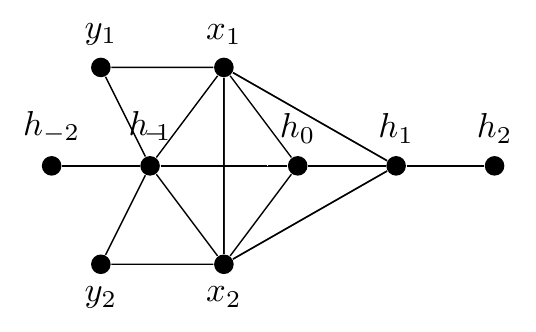} 
  }
  \,
  \subfloat[]{\label{fig:hole-3}
    \includegraphics[scale=.97]{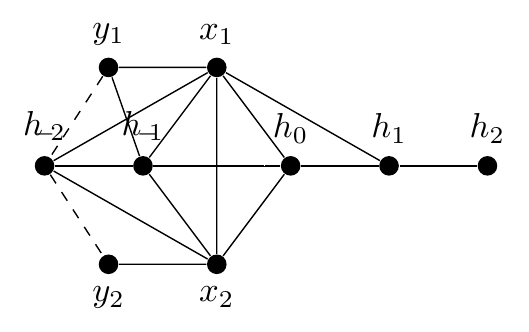} 
  }
  \,
  \subfloat[]{\label{fig:hole-4}
    \includegraphics[scale=.97]{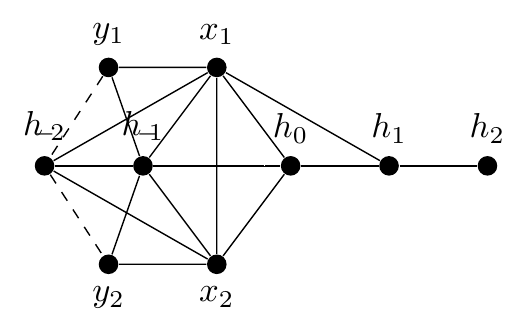} 
}
\caption{Structures used in the proof of Lemma~\ref{lem:at-with-w}
  (dashed edges may or may not exist).}
  \label{fig:negative-certificate-2}
\end{figure}


\begin{thebibliography}{10}
\small
\bibitem{cao-14-recognizing-nhcag}
Yixin Cao.
\newblock Direct and certifying recognition of normal {H}elly circular-arc
  graphs in linear time.
\newblock In Jianer Chen, John Hopcroft, and Jianxin Wang, editors, {\em
  Proceedings of the 8th International Frontiers of Algorithmics Workshop, FAW
  2014}, 2014.
\newblock To appear.

\bibitem{cao-14-almost-interval-recognition}
Yixin Cao.
\newblock Linear recognition of almost (unit) interval graphs.
\newblock arXiv:1403.1515, 2014.

\bibitem{duran-14-survey}
Guillermo Dur\'an, Luciano~N. Grippo, and Mart\'\i n~D.~Safe.
\newblock Structural results on circular-arc graphs and circle graphs: A survey
  and the main open problems.
\newblock {\em Discrete Applied Mathematics}, 164(2):427--443, 2014.
\newblock LAGOS'11: Sixth Latin American Algorithms, Graphs, and Optimization
  Symposium, Bariloche, Argentina---2011.

\bibitem{grippo-12-cag-without-dominating-triple}
Luciano~N. Grippo and Mart\'in~D. Safe.
\newblock On circular-arc graphs having a model with no three arcs covering the
  circle.
\newblock CLAIO-SBPO 2012, Rio de Janeiro - Brazil, September, 24-28 2012.
  http://www2.claiosbpo2012.iltc.br/pdf/102012.pdf.

\bibitem{hadwiger-64-combinatorial-geometry}
Hugo Hadwiger, Hans Debrunner, and Victor Klee.
\newblock {\em Combinatorial geometry in the plane}.
\newblock Athena series. Holt, Rinehart and Winston, London, 1964.

\bibitem{heggernes-07-certifying-fis}
Pinar Heggernes and Dieter Kratsch.
\newblock Linear-time certifying recognition algorithms and forbidden induced
  subgraphs.
\newblock {\em Nordic Journal of Computing}, 14(1-2):87--108, 2007.

\bibitem{kratsch-06-certifying-interval-and-permutation}
Dieter Kratsch, Ross~M. McConnell, Kurt Mehlhorn, and Jeremy~P. Spinrad.
\newblock Certifying algorithms for recognizing interval graphs and permutation
  graphs.
\newblock {\em SIAM Journal on Computing}, 36(2):326--353, 2006.
\newblock A preliminary version appeared in SODA 2003.

\bibitem{lekkerkerker-62-interval-graphs}
Cornelis~G. Lekkerkerker and J.~Ch. Boland.
\newblock Representation of a finite graph by a set of intervals on the real
  line.
\newblock {\em Fundamenta Mathematicae}, 51:45--64, 1962.

\bibitem{lin--10-clique-operator-cag}
Min~Chih Lin, Francisco~J. Soulignac, and Jayme~L. Szwarcfiter.
\newblock The clique operator on circular-arc graphs.
\newblock {\em Discrete Applied Mathematics}, 158(12):1259--1267, 2010.

\bibitem{lin-13-nhcag-and-subclasses}
Min~Chih Lin, Francisco~J. Soulignac, and Jayme~L. Szwarcfiter.
\newblock Normal {H}elly circular-arc graphs and its subclasses.
\newblock {\em Discrete Applied Mathematics}, 161(7-8):1037--1059, 2013.

\bibitem{lin-09-cag-and-subclasses}
Min~Chih Lin and Jayme~L. Szwarcfiter.
\newblock Characterizations and recognition of circular-arc graphs and
  subclasses: A survey.
\newblock {\em Discrete Mathematics}, 309(18):5618--5635, 2009.

\bibitem{lindzey-13-find-forbidden-subgraphs}
Nathan Lindzey and Ross~M. McConnell.
\newblock On finding {Tucker} submatrices and {Lekkerkerker-Boland} subgraphs.
\newblock In Andreas Brandst\"{a}dt, Klaus Jansen, and R{\"u}diger Reischuk,
  editors, {\em Revised Papers of the 39th International Workshop on
  Graph-Theoretic Concepts in Computer Science, WG 2013}, volume 8165 of {\em
  LNCS}, pages 345--357, 2013.

\bibitem{mcconnell-03-recognition-cag}
Ross~M. McConnell.
\newblock Linear-time recognition of circular-arc graphs.
\newblock {\em Algorithmica}, 37(2):93--147, 2003.

\bibitem{mcconnell-11-survey-certifying-algorithms}
Ross~M. McConnell, Kurt Mehlhorn, Stefan N{\"a}her, and Pascal Schweitzer.
\newblock Certifying algorithms.
\newblock {\em Computer Science Review}, 5(2):119--161, 2011.

\bibitem{mckee-03-restricted-cag}
Terry~A. McKee.
\newblock Restricted circular-arc graphs and clique cycles.
\newblock {\em Discrete Mathematics}, 263(1-3):221--231, 2003.

\bibitem{tarjan-84-chordal-recognition}
Robert~E. Tarjan and Mihalis Yannakakis.
\newblock Simple linear-time algorithms to test chordality of graphs, test
  acyclicity of hypergraphs, and selectively reduce acyclic hypergraphs.
\newblock {\em SIAM Journal on Computing}, 13(3):566--579, 1984.
\newblock With Addendum in the same journal, 14(1):254-255, 1985.

\bibitem{tucker-74-structures-cag}
Alan~C. Tucker.
\newblock Structure theorems for some circular-arc graphs.
\newblock {\em Discrete Mathematics}, 7(1-2):167--195, 1974.

\bibitem{tucker-75-coloring-cag}
Alan~C. Tucker.
\newblock Coloring a family of circular arcs.
\newblock {\em SIAM Journal on Applied Mathematics}, 29(3):493--502, 1975.

\end{thebibliography}
\end{document}